\documentclass[%
 reprint,
 superscriptaddress,
 nofootinbib,
 nobibnotes,
 amsmath,amssymb,
 aps,
 prl,
 floatfix,
]{revtex4-2}

\usepackage{graphicx}
\usepackage{dcolumn}
\usepackage{bm}
\usepackage[colorlinks,linkcolor=blue,urlcolor=blue,citecolor=blue]{hyperref}
\usepackage{lipsum}
\usepackage{dsfont}
\usepackage{physics}
\usepackage{enumitem}
\usepackage{tcolorbox}
\usepackage{needspace}
\usepackage{setspace}

\usepackage{amsthm} 
\usepackage{thmtools} 
\usepackage{thm-restate}

\declaretheorem[numberwithin=section]{theorem}
\declaretheorem[numberwithin=section]{lemma}

\declaretheorem[numberwithin=section]{definition}

\renewcommand{\thetheorem}{\thesection.\arabic{theorem}}
\renewcommand{\thelemma}{\thesection.\arabic{lemma}}
\renewcommand{\thecorollary}{\thesection.\arabic{corollary}}
\renewcommand{\thedefinition}{\thesection.\arabic{definition}}
\renewcommand{\theproposition}{\thesection.\arabic{proposition}}

\usepackage[page]{appendix}

\renewcommand{\appendixtocname}{List of appendices}

\makeatletter
\let\oldappendix\appendices

\g@addto@macro\tableofcontents{%
  \let\tf@toc@orig\tf@toc
}

\renewcommand{\appendices}{%
  \renewcommand{\thesection}{\Alph{section}}%

  \setcounter{section}{0}%

  \renewcommand{\thetheorem}{\thesection.\arabic{theorem}}%
  \renewcommand{\thelemma}{\thesection.\arabic{lemma}}%
  \renewcommand{\thecorollary}{\thesection.\arabic{corollary}}%
  \renewcommand{\theproposition}{\thesection.\arabic{proposition}}%
  \renewcommand{\thedefinition}{\thesection.\arabic{definition}}%

  \let\tf@toc\tf@app
  \addtocontents{app}{\protect\setcounter{tocdepth}{1}}%
  \immediate\write\@auxout{%
    \string\let\string\tf@toc\string\tf@app
  }%
  \oldappendix
}%

\g@addto@macro\endappendices{%
  \let\tf@toc\tf@toc@orig
  \immediate\write\@auxout{%
    \string\let\string\tf@toc\string\tf@toc@orig
  }%
}

\renewcommand\tableofcontents{%
    \@starttoc{toc}%
}

\newcommand{\listofappendices}{%
  \begingroup
  \newcommand{\contentsname}{\appendixtocname}
  \let\@oldstarttoc\@starttoc
  \def\@starttoc##1{\@oldstarttoc{app}}
  \tableofcontents
  \endgroup
}

\makeatother

\begin{document}

\preprint{APS/123-QED}

\title{Thermodynamic Proof of Quantumness without Structure}

\author{Florian Meier}
\email[]{florianmeier256@gmail.com}
\affiliation{Institut für theoretische Physik, Technische Universit{\"a}t Wien, 1040 Vienna, Austria}
\affiliation{Atominstitut, Technische Universit{\"a}t Wien, 1020 Vienna, Austria}
\author{Hayata Yamasaki}
\email[]{hayata.yamasaki@gmail.com}
\affiliation{Department of Computer Science, Graduate School of Information Science and Technology, The University of Tokyo, 7-3-1 Hongo, Bunkyo-ku, Tokyo 113-8656, Japan}

\begin{abstract}
Demonstrating that a machine performs genuinely quantum operations is a central challenge in quantum information processing.
Existing proofs of quantumness typically rely on computational tasks that are infeasible for classical machines under assumptions such as computational hardness, or explicit bounds on classical runtime, memory, or oracle queries.
Here we introduce an alternative paradigm: a thermodynamic proof of quantumness based on energy consumption.
Using the search problem of Yamakawa and Zhandry with a structureless uniformly random oracle, we show that any classical machine solving the task with constant success probability must generate entropy that grows exponentially with the security parameter.
By Landauer's principle, this implies an exponential lower bound on heat dissipation and hence on the energy consumption of any cyclic classical implementation.
In contrast, the corresponding quantum prover solves the task with polynomial query and gate complexity, giving polynomial energy consumption under standard assumptions for quantum implementations.
Thus, in the random-oracle setting, sufficiently low energy consumption in a successful implementation certifies quantumness through a macroscopic thermodynamic observation of energy consumption.
\end{abstract}

\maketitle

\textit{Introduction.}---%
Quantum computation can outperform classical computation in several inequivalent resource measures, including time complexity, query complexity, and communication complexity~\cite{Watrous2009,Wolf2019}. 
These advantages capture different aspects of computation and do not generally imply one another. 
For example, exponential advantages of quantum computation over classical computation are widely expected in time complexity and can be proved unconditionally in query complexity for oracle problems, whereas quantum advantages in space complexity are at most polynomial in standard settings~\cite{watrous2003complexity}. 
Identifying which resources can exhibit genuine quantum advantages, and under which assumptions, is therefore a fundamental question at the interface of quantum computation and physics.

Many exponential quantum advantages appear to require particular mathematical structures in the computational task. 
Canonical examples include period finding in Shor's algorithm~\cite{Shor1997} and Simon's problem~\cite{Simon1997}, where the quantum speedup is tied to hidden algebraic structure. 
By contrast, unstructured search admits only the quadratic speedup of Grover's algorithm~\cite{Grover1996,Bennett1997}, which is often considered insufficient for practical quantum advantage once the overheads of fault-tolerant quantum computation are included at finite scales~\cite{Babbush2021}. 
This has supported the intuition that substantial quantum advantages require highly structured problems~\cite{Aaronson2014,Aaronson2022}. 
Yamakawa and Zhandry recently challenged this intuition by constructing a search problem relative to a uniformly random oracle that admits an exponential quantum advantage in query complexity~\cite{Yamakawa2024}. 
Their result shows that exponential quantum advantage can arise even when the oracle itself has no exploitable structure.

This development raises a more physical question: does this unstructured quantum advantage persist for resources other than query complexity? 
If the separation were confined to query complexity, it could remain a formal feature of oracle access. 
If it persists for an experimentally accessible physical resource, it gives a more operational manifestation of quantumness in computation. 
Energy consumption~\cite{Meier2025} is a natural resource for this purpose~\cite{Benioff1982,Bennett1982,Feynman1986}. 
It is a central performance measure of modern computers due to increasing demands for sustainability~\cite{Auffeves2022,Jaschke2023,FellousAsiani2023,MarinGuzman2024,Campbell2026}, but it is not directly reducible to time, space, or query complexity. 
At the same time, in thermodynamics, energy is a macroscopic observable. 
It therefore occupies a special position as both a computational resource and a physical quantity accessible without directly probing microscopic quantum-mechanical properties. 
An exponential quantum advantage in energy consumption was proved in Ref.~\cite{Meier2025} for Simon's problem, but relied on a highly structured oracle and a decision task. 
It remained open whether a comparable separation exists for more general, unstructured search problems.

\begin{figure}
    \centering
    \includegraphics[width=\linewidth]{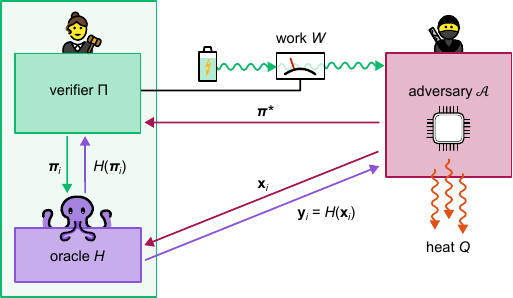}
    \caption{Thermodynamic proof of quantumness.
    The verifier samples an oracle $H$ and gives oracle access to an adversary.
    The adversary queries inputs $\mathbf x_i$, receives outputs $\mathbf y_i=H(\mathbf x_i)$, and returns a proof string $\pi$.
    The verifier checks the proof and simultaneously measures the energy consumption $W$ of the adversary.
    A correct proof with energy consumption below the classical thermodynamic lower bound certifies quantumness of the adversary's computation.}
    \label{fig:setup}
\end{figure}

In this work, we resolve this question and prove an exponential advantage in energy consumption for the unstructured Yamakawa--Zhandry search problem~\cite{Yamakawa2024}.
We formulate a \textit{thermodynamic proof of quantumness} in which a verifier observes only whether the adversary solves the task and how much energy the adversary consumes. 
For a uniformly random oracle, we prove that any classical machine solving the task with constant success probability must generate entropy of order $\lambda 2^{\lambda^c}$ for a security parameter $\lambda$ and a constant $c>0$, independently of its running time, memory size, and implementation details. 
By Landauer's principle~\cite{Landauer1961,Bennett1982}, this entropy implies an exponential lower bound on heat dissipation and hence on the energy consumption of any cyclic implementation. 
In contrast, the quantum prover of Yamakawa and Zhandry uses only polynomially many oracle queries and gates~\cite{Yamakawa2024}; under the standard setting where control and erasure costs scale polynomially with circuit size~\cite{Chiribella2021,Meier2025}, its energy consumption is polynomial in the security parameter.

Our technical contribution is to prove an entropy lower bound for all adaptive, keyed classical strategies that may run for an unbounded number of queries and stop at an input-dependent time. 
Thereby, we go beyond reinterpreting the query-complexity lower bound for the Yamakawa--Zhandry problem, and convert an averaged soundness statement for bounded-query adversaries into a lower bound on the stopping-time distribution of any successful unbounded adversary.
The proof then minimizes the entropy generated by a random oracle subject to this stopping-time constraint. 
The resulting lower bound includes explicit constants, allowing quantitative certification criteria for an experimental thermodynamic proof of quantumness. 
Thus, sufficiently low observed energy consumption in solving the task rules out every classical implementation satisfying the thermodynamic assumptions, and hence certifies that the solving device uses genuinely quantum operations.

\textit{Problem setting.}---%
We formulate the thermodynamic setting for solving the Yamakawa--Zhandry search problem~\cite{Yamakawa2024}, as illustrated in Fig.~\ref{fig:setup}; the detailed construction and the explicit parameters are presented in the Supplemental Material.
Let $\lambda$ be the security parameter. 
The oracle is a uniformly random function
\begin{align}
    H:\Sigma\to\{0,1\}^n ,
\end{align}
where the alphabet $\Sigma$ has size $|\Sigma|=2^{\lambda^{\Theta(1)}}$ and the code length satisfies $n=\Theta(\lambda)$. 
Concretely, $\Sigma$ can be represented by bit strings and is taken to be a vector space $\mathbb F_q^{\tilde m}$ over the finite field of order $q$ for suitable functions $q=q(\lambda)$ and $\tilde m=\tilde m(\lambda)$. 
Let $C_\lambda\subseteq \Sigma^n$ be the folded Reed--Solomon code~\cite{Reed1960} used in Ref.~\cite{Yamakawa2024}. 
The relevant property of this code is list recoverability, which ensures classical hardness of this problem.
For $\pi=(\mathbf x_1,\ldots,\mathbf x_n)\in\Sigma^n$, define
\begin{align}
    F_C^H(\pi)=(H_1(\mathbf x_1),\ldots,H_n(\mathbf x_n)),
\end{align}
where $H_i$ denotes the $i$th output bit of $H$. 
The verifier fixes the target string $1^n$ and accepts a proof $\pi$ with $\mathsf{Verify}^H(1^\lambda,k,\pi)=\top$ if
\begin{align}
    \pi\in C_\lambda,
    \qquad
    F_C^H(\pi)=1^n,
\end{align}
where the verification uses only polynomial time and polynomially many oracle queries.
Otherwise, the verifier returns $\perp$.
Equivalently, the adversary must find a codeword $\pi=(\mathbf x_1,\ldots,\mathbf x_n)$ such that $H_i(\mathbf x_i)=1$ for every $i$. 
The adversary may additionally receive an auxiliary key $k\in\mathcal K$ that serves as an advice string, sampled independently of $H$ from an arbitrary distribution. 
As shown in Fig.~\ref{fig:circuit}, the adversary performs internal reversible operations interleaved with oracle queries and finally outputs a proof string. 
We let
\begin{align}
    p_{\rm succ}
    =
    \Pr_{H,k}\!\left[
    \mathsf{Verify}^H(1^n,k,\mathcal A^H(1^n,k))=\top
    \right]
\end{align}
denote the success probability of an adversary with proof string $\pi=\mathcal A^H(1^n,k)$ over the random oracle and the key. 
Throughout the lower bound, we require
\begin{align}
\label{eq:Pr_Success_Geq_Delta}
    p_{\rm succ}\geq \Delta
\end{align}
for a constant $\Delta>0$ independent of $\lambda$.
For a constant $0<c<1$, this problem supports a proof of quantumness that is classically sound against $Q(\lambda)=2^{\lambda^c}$ oracle queries with soundness error $\varepsilon(\lambda)=2^{-\Omega(\lambda)}$, while admitting a quantum polynomial-time prover~\cite{Yamakawa2024}.

\begin{figure}
    \centering
    \includegraphics[width=\linewidth]{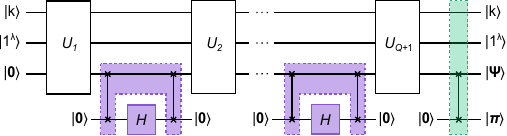}
    \caption{Adaptive oracle computation performed by the adversary.
    Given the key $k$ and the input $1^n$, the adversary applies internal reversible operations $U_1,\ldots,U_{Q+1}$ interleaved with $Q$ oracle queries.
    The number of queries may depend on previous oracle outputs in the classical case, and the oracle may be queried coherently in the quantum case.
    The final register contains the proof string $\pi$ returned to the verifier, and the adversary's internal memory state $\rho_{|H,k}$.
    The irreversible erasure step used to close the thermodynamic cycle is not shown.}
    \label{fig:circuit}
\end{figure}

We now specify the thermodynamic assumptions. 
The adversary is run repeatedly, and after each run its internal memory is reset to a standard initial state so that the same physical device can be used again. 
This forms a thermodynamic cycle.
For a fixed oracle and key, let $\rho_{|H,k}$ be the adversary's final memory state before erasure. 
Because the verifier samples $H$ and $k$, the state to be erased is
\begin{align}
\label{eq:rho_sumHk}
    \rho=\sum_{H,k}\Pr[H]\Pr[k]\rho_{|H,k}.
\end{align}
For a classical adversary, after removing redundant repeated queries without changing its success probability, the final memory can be taken to contain a classical transcript
\begin{align}
    (\mathbf X,\mathbf Y)
    =
    (\mathbf x_1,\mathbf x_2,\ldots;\mathbf y_1,\mathbf y_2,\ldots),
    \qquad
    \mathbf y_i=H(\mathbf x_i),
\end{align}
where the inputs may depend adaptively on previous outputs and on $k$. 
Then $\rho$ is diagonal in the standard basis, and its von Neumann entropy equals the Shannon entropy
\begin{align}
\label{eq:classical_entropy}
    S[\rho]=H[\mathbf X,\mathbf Y]
    =
    -\sum_{\mathbf X,\mathbf Y}
    \Pr[\mathbf X,\mathbf Y]\log[\Pr[\mathbf X,\mathbf Y]] .
\end{align}
Landauer's principle gives the heat dissipated into a bath at temperature $T$ when this state is reset:
\begin{align}
\label{eq:Q_geq_kBTS}
    Q_{\rm heat}\geq k_B T S[\rho].
\end{align}
For a closed computational cycle in which the device returns to its initial state, the consumed energy $W$ is at least the dissipated heat~\cite{Meier2025}. 
Thus, any lower bound on $S[\rho]$ gives an implementation-independent lower bound on energy consumption, apart from the explicit thermodynamic assumption of cyclic reuse and erasure.

\textit{Classical entropy lower bound.}---%
Our main classical result is the following: for the Yamakawa--Zhandry problem, any classical adversary achieving Eq.~\eqref{eq:Pr_Success_Geq_Delta} has
\begin{align}
\label{eq:main_entropy_bound}
    S[\rho]\geq \Theta(\lambda 2^{\lambda^c}),
\end{align}
for a constant $c>0$ determined by the code parameters. 
More explicitly, in the Supplemental Material, we prove a finite-$\lambda$ bound of the form
\begin{align}
\label{eq:explicit_entropy_schematic}
    S[\rho]&\geq
    (\Delta-\varepsilon(\lambda))\times\nonumber\\
    &\quad\left(
    n\,2^{\lambda^c}\log[2]
    +
    \log\!\left[
    (\Delta-\varepsilon(\lambda))(1-2^{-n})
    \right]
    \right),
\end{align}
with $\varepsilon(\lambda)=2^{-\Omega(\lambda)}$ for the folded Reed--Solomon code parameters, where $\log$ is the natural logarithm. 
Equation~\eqref{eq:main_entropy_bound} follows because $n=\Theta(\lambda)$ and $\Delta-\varepsilon(\lambda)$ is bounded below by a positive constant for fixed $\Delta>0$ and sufficiently large $\lambda$.

We sketch the proof idea while the details are presented in the Supplemental Material. 
Fix a key $k$ and consider a deterministic adaptive classical strategy, which is, without loss of generality, after absorbing the adversary's private randomness into the key. 
Let $\mathcal Y_{Q,{\rm S}}^k$ be the set of output transcripts for which the adversary stops after the $Q$th nonredundant query, and let $Y_{Q,{\rm S}}^k=|\mathcal Y_{Q,{\rm S}}^k|$. 
Since the oracle is uniformly random and the key is independent of the oracle, the output transcript conditioned on stopping after $Q$ queries is uniform over $\mathcal Y_{Q,{\rm S}}^k$. 
Moreover, the number $Q^{\rm S}$ of queries in this case satisfies
\begin{align}
\label{eq:PrQS|k_main}
    \Pr[Q^{\rm S}=Q|k]=\frac{Y_{Q,{\rm S}}^k}{2^{nQ}}.
\end{align}
It follows that the transcript entropy conditioned on $k$ satisfies
\begin{align}
\label{eq:HkY_main}
    H_k[\mathbf Y]
    \geq
    \sum_{Q\geq 1}
    \frac{Y_{Q,{\rm S}}^k}{2^{nQ}}
    \log[Y_{Q,{\rm S}}^k].
\end{align}
This lemma is a purely information-theoretic statement about adaptive sampling from a random oracle; it does not use the algebraic structure of the code.

The second step is to minimize the right-hand side of Eq.~\eqref{eq:HkY_main} under a stopping-time constraint. 
For each fixed key $k$, define
\begin{align}
    q_k=\Pr[Q^{\rm S}\geq Q^*|k].
\end{align}
Applying the fixed-key entropy bound with $q=q_k$ gives
\begin{align}
\label{eq:HkY_main_lower}
    H_k[\mathbf Y]
    \geq
    q_k n Q^*\log[2]
    +
    q_k\log[q_k(1-2^{-n})].
\end{align}
Thus, forcing a classical algorithm to make many random-oracle queries with nonzero probability forces it to store a large amount of output entropy before erasure.

It remains to obtain a stopping-time constraint for any adversary that succeeds with probability at least $\Delta$. 
Here we use the averaged soundness statement over the random oracle and the key, rather than a fixed-key soundness statement.
If a classical adversary succeeds with probability at least $\Delta$ but stops before $Q^*=2^{\lambda^c}$ queries except with probability smaller than $\Delta-\varepsilon(\lambda)$, then truncating it before the $Q^*$th query would give a bounded-query adversary with success probability larger than the Yamakawa--Zhandry soundness error. 
Therefore,
\begin{align}
\label{eq:avg_min_queries}
    \Pr_{H,k}[Q^{\rm S}\geq Q^*]\geq \Delta-\varepsilon(\lambda).
\end{align}
Averaging Eq.~\eqref{eq:HkY_main_lower} over $k$ and applying Jensen's inequality to the resulting convex function of $q_k$ yields the averaged finite-size bound in Eq.~\eqref{eq:explicit_entropy_schematic}. 
Finally, using
\begin{align}
    S[\rho]=H[\mathbf X,\mathbf Y]\geq H[\mathbf Y]\geq H[\mathbf Y|k],
\end{align}
gives Eq.~\eqref{eq:main_entropy_bound}. 
The argument allows the classical adversary arbitrary computation time, arbitrary memory, arbitrary adaptive control, and an unbounded stopping time. 
The lower bound arises only from the entropy of the random-oracle information that a successful classical strategy must acquire.

Landauer's principle with Eq.~\eqref{eq:main_entropy_bound} and $W_{\rm C}\geq Q_{\rm heat}$ gives the classical energy consumption lower bound
\begin{align}
\label{eq:classical_energy_bound}
    W_{\rm C}
    \geq
    k_B T\,\Theta(\lambda 2^{\lambda^c}),
\end{align}
for any classical thermodynamic-cycle implementation satisfying Eq.~\eqref{eq:Pr_Success_Geq_Delta}~\cite{Meier2025}. 
This is the thermodynamic soundness statement: a machine that solves the verification task with constant success probability and energy consumption below the bound cannot be classical under the stated assumptions.

\begin{figure}
    \centering
    \includegraphics[width=\linewidth]{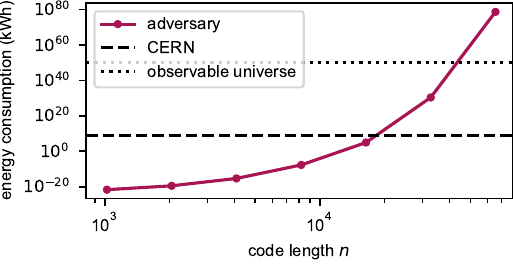}
    \caption
    {Classical energy-consumption lower bound for a representative parameter choice at $T=300\,{\rm K}$.
    The dots show $k_BTS$ using the explicit entropy bound from the Supplemental Material.
    The comparison lines indicate the annual energy consumption of CERN~\cite{CERN:EnvironmentReport2021-22} and the annual luminosity of the observable universe~\cite{enwiki:1319956122}.}
    \label{fig:energy}
\end{figure}

\textit{Quantum achievability.}---%
The separation is meaningful only if the same verification task is achievable by a quantum machine with polynomial energy consumption. 
Yamakawa and Zhandry construct a quantum polynomial-time prover for the above search problem whose failure probability is negligible and whose number of oracle queries is polynomial in $\lambda$~\cite{Yamakawa2024}. 
Let $D$ and $N_{\rm q}$ denote the depth and width of an implementation of this prover, including the reversible implementation of the oracle. 
For the Yamakawa--Zhandry prover, both satisfy
\begin{align}
    D={\rm poly}(\lambda),
    \qquad
    N_{\rm q}={\rm poly}(\lambda).
\end{align}
The fundamental energy consumption for erasure is at most polynomial because the final memory size is polynomial. 
Other costs, including finite-time erasure~\cite{VanVu2022,Rolandi2023,Taranto2023}, state preparation~\cite{Xuereb2025}, gate control~\cite{Chiribella2021,Chiribella2022,Castellano2025,Aaberg2014,Woods2024} and dissipative losses in the classical control hardware~\cite{Xuereb2023,FellousAsiani2023}, may dominate in real devices~\cite{FellousAsiani2023,Auffeves2022}. 
However, under the standard physical assumption that the total control and erasure cost scales polynomially with the implemented circuit size, as in the model of Ref.~\cite{Meier2025}, the total quantum energy consumption satisfies
\begin{align}
\label{eq:quantum_energy_bound}
    W_{\rm Q}=O({\rm poly}(\lambda)).
\end{align}
Equations~\eqref{eq:classical_energy_bound} and \eqref{eq:quantum_energy_bound} give an exponential separation in energy consumption between classical and quantum solutions of the same unstructured search problem.

\textit{Experimental instantiation.}---%
A truly uniformly random oracle of the required size cannot be implemented explicitly for large $\lambda$. 
As in cryptographic applications of the random-oracle model~\cite{Bellare1993}, the natural experimental instantiation is to replace the random oracle by a pseudorandom function, for example the SHA-2 hash function, as suggested in Ref.~\cite{Yamakawa2024}. 
To be queried by the quantum prover, this function must be implemented reversibly and coherently so that it acts on superposition inputs, as in the quantum random-oracle model~\cite{Boneh2011}. 
This does not require the verifier to possess a universal quantum computer. 
The verifier needs to instantiate and specify the oracle, check the returned classical proof, and measure the adversary's energy consumption; the ability to query the oracle coherently is required of the adversary's device.

The proposed setup for experimental demonstration consists of many independent rounds. 
In each round, the verifier chooses a fresh pseudorandom oracle instance and key, the adversary attempts to produce a proof, and the verifier records both success or failure and the adversary's energy consumption, for example, by monitoring the electrical energy supplied to the device over the full computational cycle. 
The protocol accepts the device only if the empirical success frequency is at least $\Delta+\eta$ for a fixed margin $\eta>0$ and the average energy consumption per round lies below the classical lower bound. 
If the true success probability were smaller than $\Delta$, the probability of observing empirical success frequency at least $\Delta+\eta$ is bounded by a Chernoff inequality~\cite{Chernoff1952},
\begin{align}
    \Pr[p_{\rm emp}\geq\Delta+\eta\mid p_{\rm succ}<\Delta]
    \leq
    e^{-\Omega(M)},
\end{align}
where $M$ is the number of rounds. 
Thus, the success condition can be certified statistically, while the measured energy consumption is compared against the explicit criteria obtained by evaluating the finite-size lower bound in Eq.~\eqref{eq:explicit_entropy_schematic}, which we plot in Fig.~\ref{fig:energy}. 
The proof would be unconditional in the ideal random-oracle model; with the proposed instantiations, the conclusion becomes conditional on the security assumption for the hash function, as in standard cryptographic practice~\cite{Bellare1993,Boneh2011,Canetti2004}.

\textit{Discussion.}---%
Thermodynamics describes physical systems through coarse-grained variables that are usually ignorant of microscopic information. 
At microscopic scales, however, quantum theory governs the physical laws underlying computation. 
It is therefore natural to ask whether a macroscopic thermodynamic observation can certify quantumness of the underlying process~\cite{OliveiraJunior2025,OliveiraJunior2026,Macedo2026}. 
Entropy production, heat capacity and heat flow can reveal quantum features in specific systems~\cite{Esposito2010,Reeb2014,Biswas2025,Comar2025,OliveiraJunior2025,Khandelwal2025a,Bourgeois2025,Xue2026,Xuereb2026}, but such witnesses are typically model-dependent or tailored to particular physical platforms.

By contrast, we have shown that, in a cryptographic oracle setting, energy consumption can provide a model-agnostic proof of quantumness. 
Using the unstructured search problem of Yamakawa and Zhandry~\cite{Yamakawa2024}, we proved that any classical machine solving the task with constant success probability must acquire and erase an exponentially large amount of random-oracle information. 
Landauer's principle converts this entropy requirement into an exponential lower bound on energy consumption. 
A quantum prover, by contrast, solves the same task with polynomial query and circuit complexity, and hence with polynomial energy consumption under standard scaling assumptions for implementation costs. 
In our formulation, the erasure of the adversary's memory after each round is part of the thermodynamic model: the lower bound applies to cyclic implementations that return the device to a standard initial state before reuse.
Developing a stronger verification framework that forces, or operationally certifies, such erasure for an arbitrary adversary is an important direction for future work.

Our result shows that the unstructured quantum advantage of Ref.~\cite{Yamakawa2024} is not merely a query-complexity phenomenon: it can be witnessed as an exponential separation in a thermodynamic quantity. 
Consequently, sufficiently low energy consumption in a successful implementation serves as a thermodynamic proof of quantumness, with explicit criteria for future proof-of-principle experimental demonstrations as illustrated in Fig.~\ref{fig:energy}.

\textit{Acknowledgments.---}%
The authors thank Natsuto Isogai, Takashi Yamakawa, Sebastian Haslebacher, and Marcus Huber for discussions.
FM is co-funded by the European Union through 
the ERC Consolidator Grant Cocoquest (Grant Agreement No.~101043705),
the ERC Synergy Grant SuperWave (Grant Agreement No.~101071882)
and by the Austrian Science Fund FWF (Grant DOI:~10.55776/COE1).
Views and opinions expressed are however those of the authors only and do not necessarily reflect those of the European Union, REA or UKRI.
Neither the European Union nor UKRI can be held responsible for them.
HY was supported by JST PRESTO Grant Number JPMJPR23FC, JST CREST Grant Number JPMJCR25I5, JST Moonshot R\&D Grant Number JPMJMS256J, and Faculty Research Funding from Google Quantum AI\@. 

\textit{Data availability statement.---}%
No data was generated in this study.

\bibliography{bibliography}

\clearpage\newpage

\onecolumngrid

\setcounter{secnumdepth}{2}
\section*{Supplemental Material}
\appendix

\listofappendices

\begin{appendices}

The Supplemental Material is organized as follows.
In Sec.~\ref{SM:problem_summary}, we recall the Yamakawa--Zhandry problem setup, including the code-theoretic definitions, the relevant folded Reed--Solomon parameters, and the corresponding proof-of-quantumness construction.
In Sec.~\ref{SM:proofs}, we prove the classical entropy lower bound by first analyzing adaptive random-oracle transcripts and then combining the resulting entropy bounds with the averaged soundness of the Yamakawa--Zhandry protocol.
In Sec.~\ref{SM:chernoff_proof}, we give the Chernoff-bound calculation used to justify the empirical success-rate test in the proposed experimental demonstration.

\section{\label{SM:problem_summary}Details of the problem setup}

In this appendix, we review the key notions needed to formulate the problem introduced by Yamakawa and Zhandry~\cite{Yamakawa2024}.
The problem is defined using a uniformly random oracle
\begin{align}
    H:\Sigma\to\{0,1\}^n
\end{align}
on a finite domain $\Sigma$.
The adversary's goal is to find an input sequence $\pi=(\mathbf x_1,\ldots,\mathbf x_n)$ such that
\begin{align}
    H_i(\mathbf x_i)=1
    \quad\text{for all } i=1,\ldots,n,
    \qquad
    \pi\in C,
\end{align}
where $H_i$ denotes the $i$th output bit of $H$, and $C\subseteq\Sigma^n$ is a specified code.
The code $C$ is chosen so that finding such a sequence is exponentially hard for classical adversaries with a bounded number of oracle queries.

\subsection{Definition of a suitable code}

We first recall the definitions needed to specify the code $C$.

\begin{definition}[Code~\cite{Yamakawa2024}]
A code of length $n\in\mathbb N$ over a finite alphabet $\Sigma$ is a subset $C\subseteq\Sigma^n$.
That is, $C$ is a set of $n$-tuples over $\Sigma$.
\end{definition}

The alphabet $\Sigma$ can be an arbitrary finite set.
When the alphabet has an algebraic structure, one can define linear and folded linear codes.

\begin{definition}[Linear code~\cite{Yamakawa2024}]
A linear code of length $n$ over a finite field $\mathbb F_q$ is a code $C\subseteq\mathbb F_q^n$ that is an $\mathbb F_q$-linear subspace.
That is, for any $\vec c_1,\vec c_2\in C$ and any $\alpha\in\mathbb F_q$, we have
$\vec c_1+\alpha \vec c_2\in C$ .
\end{definition}

\begin{definition}[Folded linear code~\cite{Yamakawa2024}]
A folded linear code of length $n$ over the alphabet $\Sigma=\mathbb F_q^r$ is a code $C\subseteq\Sigma^n$ that is an $\mathbb F_q$-linear subspace of $(\mathbb F_q^r)^n$.
Equivalently, it can be obtained by grouping the coordinates of an $\mathbb F_q$-linear code into blocks.
More explicitly, let $\widetilde C\subseteq\mathbb F_q^N$ be a linear code, and let $r$ be a positive integer dividing $N$.
The $r$-folded version of $\widetilde C$ is
\begin{align}
    \widetilde C^{(r)}
    :=
    \Big\{
    \big(
    (\mathbf x_1,\ldots,\mathbf x_r),
    (\mathbf x_{r+1},\ldots,\mathbf x_{2r}),
    \ldots,
    (\mathbf x_{N-r+1},\ldots,\mathbf x_N)
    \big)
    :
    \mathbf x\in \widetilde C
    \Big\}
    \subseteq
    (\mathbb F_q^r)^{N/r}.
\end{align}
\end{definition}

A folded linear code is therefore linear when regarded as an $\mathbb F_q$-linear subspace after identifying $(\mathbb F_q^r)^n$ with $\mathbb F_q^{rn}$.
For a linear code, we also use the following dual code.

\begin{definition}[Dual code~\cite{Yamakawa2024}]
For a linear code $C\subseteq\mathbb F_q^n$, its dual is the orthogonal complement with respect to the standard inner product:
\begin{align}
    C^\perp
    =
    \{\mathbf z\in\mathbb F_q^n:
    \mathbf x\cdot \mathbf z=0
    \text{ for all } \mathbf x\in C\}.
\end{align}
\end{definition}

If $C$ has dimension $k$, then its dual $C^\perp$ has dimension $n-k$.
The list recoverability is central to the choice of codes suitable for the Yamakawa--Zhandry problem.

\begin{definition}[List recovery~\cite{Yamakawa2024}]
\label{def:list_recovery}
A code $C\subseteq\Sigma^n$ is $(\zeta,\ell,L)$-list recoverable if, for any subsets $S_i\subseteq\Sigma$ with $|S_i|\leq \ell$ for all $i=1,\ldots,n$, we have
\begin{align}
\label{eq:list_recoverability_definition}
    \left|
    \left\{
    \vec{\mathbf x}=(\mathbf x_1,\ldots,\mathbf x_n)\in C:
    \left|\{i:\mathbf x_i\in S_i\}\right|\geq (1-\zeta)n
    \right\}
    \right|
    \leq L .
\end{align}
\end{definition}

The intuition is as follows.
Given lists of possible letters $S_1,\ldots,S_n$, list recoverability bounds the number of codewords whose coordinates lie in the corresponding lists for at least a fraction $1-\zeta$ of positions.
Thus, for arbitrary lists of size at most $\ell$, only at most $L$ codewords in $C$ can be recovered in this approximate sense.

An equivalent formulation uses a distance from a sequence of lists~\cite{Tamo2024}.
Let $S=(S_1,\ldots,S_n)$ and define
\begin{align}
    \operatorname{dist}(\mathbf x,S)
    =
    \left|\{i:\mathbf x_i\notin S_i\}\right|.
\end{align}
The distance is zero exactly when every coordinate of $\mathbf x$ lies in the corresponding list.
More generally, it counts the number of coordinates not contained in the lists.
The condition in Eq.~\eqref{eq:list_recoverability_definition} can be rewritten as
\begin{align}
    \left|
    \left\{
    \vec{\mathbf x}\in C:
    \operatorname{dist}(\mathbf x,S)\leq \zeta n
    \right\}
    \right|
    \leq L ,
\end{align}
up to the harmless convention of using $<$ or $\leq$ when $\zeta n$ is not an integer.
Thus, $C$ is $(\zeta,\ell,L)$-list recoverable if the number of codewords in $C$ within distance at most $\zeta n$ from any list sequence $S$ is at most $L$.

The codes used in the Yamakawa--Zhandry construction satisfy the following properties.

\begin{lemma}[Suitable codes~\cite{Yamakawa2024}]
\label{lemma:suitable_codes}
Let $0<c<c'<1$ be arbitrary constants.
There exists an explicit family $\{C_\lambda\}_{\lambda\in\mathbb N}$ of folded linear codes over alphabets $\Sigma=\mathbb F_q^r$ of length $n$, where
\begin{align}
    |\Sigma|=2^{\lambda^{\Theta(1)}},
    \qquad
    n=\Theta(\lambda),
    \qquad
    |C_\lambda|\geq 2^{n+\lambda},
\end{align}
such that the following hold:
\begin{enumerate}
    \item $C_\lambda$ is $(\zeta,\ell,L)$-list recoverable with
    \begin{align}
        \zeta=\Omega(1),
        \qquad
        \ell=2^{\lambda^c},
        \qquad
        L=2^{\widetilde O(\lambda^{c'})}.
    \end{align}

    \item There exists an efficient deterministic decoding algorithm $\mathsf{Decode}_{C_\lambda^\perp}$ for $C_\lambda^\perp$ with the following property.
    Let $\mathcal D$ be the distribution over $\Sigma$ that outputs $\mathbf 0$ with probability $1/2$ and otherwise outputs a uniformly random element of $\Sigma\setminus\{\mathbf 0\}$.
    Then
    \begin{align}
        \Pr_{\mathbf e\xleftarrow{\$}\mathcal D^n}
        \left[
        \forall \mathbf x\in C_\lambda^\perp:
        \mathsf{Decode}_{C_\lambda^\perp}(\mathbf x+\mathbf e)=\mathbf x
        \right]
        =
        1-2^{-\Omega(\lambda)} .
    \end{align}
\end{enumerate}
Here, $\mathbf x\xleftarrow{\$}\mathcal A$ denotes the output of a randomized classical or quantum algorithm $\mathcal A$, and $\lambda$ is the security parameter.
\end{lemma}

Yamakawa and Zhandry argue that folded Reed--Solomon codes~\cite{Guruswami1999,Rudra2007,Guruswami2008} are the only known codes satisfying all the required properties in their construction.
For the present work, only the list-recoverability condition in Lemma~\ref{lemma:suitable_codes} is used directly.
Further details on list recoverability are given in Sec.~\ref{sec:details_list_recoverability}.
Lemma~\ref{lemma:suitable_codes} in Ref.~\cite{Yamakawa2024} also includes a third condition, which is not directly needed for our lower bound.
The second condition above is required for the quantum algorithm in the Yamakawa--Zhandry problem:
errors $\mathbf e$ affecting the dual code $C_\lambda^\perp$ must be corrected efficiently.
That is, the decoding algorithm $\mathsf{Decode}_{C_\lambda^\perp}$ should recover $\mathbf x$ from $\mathbf x+\mathbf e$ for all $\mathbf x\in C_\lambda^\perp$ with overwhelmingly high probability over \smash{$\mathbf e\xleftarrow{\$}\mathcal D^n$}.

\subsection{Summary of constant factors for list-recoverable codes}
\label{sec:details_list_recoverability}

To make explicit the constant factors appearing in the classical lower bound, we review in more detail the list-recoverability criterion in Lemma~\ref{lemma:suitable_codes} and Definition~\ref{def:list_recovery}.

Following the construction in Sec.~4.3.2 of Ref.~\cite{Yamakawa2024}, the folded Reed--Solomon (FRS) codes are defined over the finite field $\mathbb F_q$, where
\begin{align}
    q=2^{2\lfloor\log_2\lambda\rfloor}
\end{align}
is a prime power.
We then define
\begin{align}
    m=2^{\lfloor\log_2\lambda\rfloor}+1,
    \qquad
    n=2^{\lfloor\log_2\lambda\rfloor}-1.
\end{align}
The FRS code $C_\lambda$ is a code of length $n$ over the alphabet $\Sigma=\mathbb F_q^m$, defined by
\begin{align}
    C_\lambda
    =
    \Big\{
    \big(
    (p(\gamma),\ldots,p(\gamma^m)),
    (p(\gamma^{m+1}),\ldots,p(\gamma^{2m})),
    \ldots,
    (p(\gamma^{(n-1)m+1}),\ldots,p(\gamma^{mn}))
    \big)
    :
    p\in\mathbb F_q[X],\ \deg[p]\leq k
    \Big\},
\end{align}
where $\gamma\in\mathbb F_q^*$ is a generator of the multiplicative group of $\mathbb F_q$, and $\mathbb F_q[X]$ is the polynomial ring over $\mathbb{F}_q$, and $\deg[p]$ is the degree of the polynomial $p$.
The product of the folding parameter and the folded code length is
\begin{align}
    mn
    =
    \left(2^{\lfloor\log_2\lambda\rfloor}+1\right)
    \left(2^{\lfloor\log_2\lambda\rfloor}-1\right)
    =
    2^{2\lfloor\log_2\lambda\rfloor}-1
    =
    q-1.
\end{align}
Thus, $C_\lambda$ is the $m$-folded Reed--Solomon code, often denoted by
\begin{align}
    {\rm RS}_{\mathbb F_q,\gamma,k}^{(m)},
\end{align}
of length $n$ over the alphabet $\Sigma=\mathbb F_q^m$, or equivalently a Reed--Solomon code of length $N=mn$ over the field $\mathbb F_q$ before folding.

Given arbitrary constants $0<c<c'<1$, the code $C_\lambda$ is $(\zeta,\ell,L)$-list recoverable for the following parameter choices~\cite{Yamakawa2024,Rudra2007}:
\begin{itemize}
    \item $0<\zeta<1-\alpha$ is an arbitrary constant, where $5/6<\alpha<1$ and the degree parameter of the FRS code is chosen as $k=\lfloor\alpha N\rfloor$.
    \item $\ell=2^{\lambda^c}$, as in Lemma~\ref{lemma:suitable_codes}.
    \item $L\leq \lambda^{2\lambda^{c'}}$, by Lemma~4.3 of Ref.~\cite{Yamakawa2024}.
\end{itemize}
The parameter $L$ can be tightened further using recent results from Ref.~\cite{Tamo2024}.

\subsection{\label{SM:YZ_construction}Yamakawa--Zhandry's construction for a proof of quantumness}

To construct the proof of quantumness, we work in a random-oracle model.
In the classical random-oracle model (CROM)~\cite{Bellare1993}, one considers a uniformly random function
\begin{align}
    H:\{0,1\}^a\to\{0,1\}^b,
\end{align}
where the input length $a=a(\lambda)$ and output length $b=b(\lambda)$ are functions of the security parameter $\lambda$.
A classical adversary can query an input bit string $\mathbf x\in\{0,1\}^a$ and receives the corresponding output $\mathbf y=H(\mathbf x)$.
The quantum random-oracle model (QROM)~\cite{Boneh2011} generalizes this model by allowing coherent queries in superposition.
The corresponding oracle unitary is written as
\begin{align}
    O_H\ket{\mathbf x,\mathbf z}
    =
    \ket{\mathbf x,\mathbf z\oplus H(\mathbf x)},
\end{align}
and hence, by linearity,
\begin{align}
    O_H
    \sum_{\mathbf x,\mathbf z}\alpha_{\mathbf x,\mathbf z}
    \ket{\mathbf x,\mathbf z}
    =
    \sum_{\mathbf x,\mathbf z}\alpha_{\mathbf x,\mathbf z}
    \ket{\mathbf x,\mathbf z\oplus H(\mathbf x)}.
\end{align}

For the Yamakawa--Zhandry problem, let $\{C_\lambda\}_\lambda$ be a family of folded linear codes of length $n$ over the alphabet $\Sigma=\mathbb F_q^m$.
Assume that this code family satisfies the requirements of Lemma~\ref{lemma:suitable_codes} for some constants $0<c<c'<1$; in particular, one may take $C_\lambda$ to be the appropriate folded Reed--Solomon codes.
Let
\begin{align}
    H:\Sigma\to\{0,1\}^n
\end{align}
be a random oracle, chosen uniformly from all functions $\Sigma\to\{0,1\}^n$.
For $1\leq i\leq n$, let $H_i$ denote the $i$th output bit of $H$.

A proof-of-quantumness protocol $\Pi=(\mathsf{Prove},\mathsf{Verify})$ is defined as follows.
\begin{itemize}
    \item $\mathsf{Prove}$ is a quantum polynomial-time (QPT) algorithm that takes a key $k$ and the security parameter $1^\lambda$, makes at most ${\rm poly}(\lambda)$ queries to the oracle, and returns a classical proof $\pi$.
    For the detailed construction of $\mathsf{Prove}$, we refer to Ref.~\cite{Yamakawa2024}.

    \item $\mathsf{Verify}$ is a verifier with oracle access.
    Given $1^\lambda$, $k$, and a candidate proof
    \begin{align}
        \pi=\mathbf x=(\mathbf x_1,\ldots,\mathbf x_n)\in\Sigma^n,
    \end{align}
    the verifier outputs $\top$ if
    \begin{align}
        \mathbf x\in C_\lambda
        \quad\text{and}\quad
        H_i(\mathbf x_i)=1
        \quad\text{for all } i=1,\ldots,n.
    \end{align}
    Otherwise, it outputs $\perp$.
\end{itemize}

The following two definitions state when a proof-of-quantumness protocol $\Pi$ is correct and sound.
Correctness means that there exists an efficient quantum prover that makes the verifier reject only with negligible probability.
Soundness means that any classical adversary making a bounded number of oracle queries can make the verifier accept only with small probability.

\begin{definition}[Correctness~\cite{Yamakawa2024}]
\label{def:correctness}
A proof-of-quantumness protocol $\Pi$ satisfies correctness if
\begin{align}
    \Pr_{H,k}
    \left[
    \mathsf{Verify}^H(1^\lambda,k,\pi)=\perp
    :
    \pi\xleftarrow{\$}\mathsf{Prove}^H(1^\lambda,k)
    \right]
    \leq
    \mathsf{negl}(\lambda).
\end{align}
\end{definition}

Here, $\mathsf{negl}(\lambda)$ denotes a negligible function that decays faster than any polynomial; i.e., for every constant $\alpha>0$,
\begin{align}
    \mathsf{negl}(\lambda)=o(\lambda^{-\alpha})
\end{align}
as $\lambda\to\infty$.

\begin{definition}[Soundness~\cite{Yamakawa2024}]
\label{def:soundness}
A proof-of-quantumness protocol $\Pi$ is $(Q(\lambda),\varepsilon(\lambda))$-sound in the classical random-oracle model if, for any classical adversary $\mathcal A$ with unbounded computational time that makes at most $Q(\lambda)$ classical queries to $H$, we have
\begin{align}
    \Pr_{H,k}
    \left[
    \mathsf{Verify}^H(1^\lambda,k,\pi^*)=\top
    :
    \pi^*\xleftarrow{\$}\mathcal A^H(1^\lambda,k)
    \right]
    \leq
    \varepsilon(\lambda).
\end{align}
\end{definition}

\begin{lemma}[Soundness of the Yamakawa--Zhandry verifier~\cite{Yamakawa2024}]
\label{lemma:soundness}
For the above construction, the verifier satisfies $(2^{\lambda^c},\varepsilon(\lambda))$-soundness in the CROM, with $\varepsilon(\lambda)=L2^{-\lfloor \zeta n \rfloor} = e^{-\Omega(\lambda)}$.
\end{lemma}

\section{\label{SM:proofs}Proofs for the classical lower bound}

In this appendix, we prove the three main ingredients used to derive the classical entropy lower bound.
We first prove Eq.~\eqref{eq:HkY_main} from the main text in Sec.~\ref{SM:proof1}.
We then prove Eq.~\eqref{eq:HkY_main_lower} from the main text in Sec.~\ref{SM:proof2}.
Finally, in Sec.~\ref{SM:proof_lowerbound}, we combine these ingredients to obtain the finite lower bound in Eq.~\eqref{eq:explicit_entropy_schematic} of the main text.

\subsection{\label{SM:proof1}Proof of Eq.~\eqref{eq:HkY_main} in the main text}

We begin by restating Eq.~\eqref{eq:HkY_main} from the main text as a lemma.

\begin{lemma}[Entropy of a keyed deterministic adversary]
\label{lemma:entropy_general}
Let $H:\Sigma\to\{0,1\}^n$ be a uniformly random oracle, and let $\mathcal A$ be a keyed deterministic classical adversary querying the oracle adaptively.
Without loss of generality, assume that repeated queries to the same oracle input are removed, so that the query transcript contains only nonredundant queries.
For a fixed key $k$, let $Y_{Q,{\rm S}}^k$ be the number of output transcripts
\begin{align}
    (\mathbf y_1,\ldots,\mathbf y_Q)\in(\{0,1\}^n)^Q
\end{align}
for which the adversary stops after the $Q$th query.
Then the entropy of the output transcript conditioned on the key $k$ is lower bounded by
\begin{align}
    H_k[\mathbf Y]
    \geq
    \sum_{Q\geq 1}
    \frac{Y_{Q,{\rm S}}^k}{2^{nQ}}
    \log Y_{Q,{\rm S}}^k .
\end{align}
\end{lemma}

Before proving Lemma~\ref{lemma:entropy_general}, we introduce the notation used in the proof.
Without loss of generality, we assume that the adversary makes at least one nonredundant query.
If the adversary makes no query, the output transcript is empty and contributes no entropy.
As stated in Lemma~\ref{lemma:entropy_general}, repeated queries to the same oracle input may be removed without changing the adversary's transcript information relevant to the proof, and we count only nonredundant queries.

For every $Q\geq 1$, let the stopping rule
\begin{align}
    f_Q(\mathbf y_1,\ldots,\mathbf y_Q,k)
    =
    \begin{cases}
    {\rm S}, & \text{if the adversary stops after the $Q$th query,}\\
    {\rm C}, & \text{if the adversary continues after the $Q$th query}
    \end{cases}
\end{align}
determine whether the adversary stops or continues, given the key $k$ and the output transcript $(\mathbf y_1,\ldots,\mathbf y_Q)$ obtained so far.
This function defines sets of output strings for which the algorithm stops or continues after $Q$ queries.
See Fig.~\ref{fig:tree} for a visualization.

For $Q=1$, define
\begin{align}
    \mathcal Y_{1,{\rm S}}^k
    &=
    \left\{
    \mathbf y_1\in\{0,1\}^n:
    f_1(\mathbf y_1,k)={\rm S}
    \right\}
\end{align}
as the set of one-query output transcripts for which the adversary stops after the first query.
Its complement is
\begin{align}
    \mathcal Y_{1,{\rm C}}^k
    =
    \{0,1\}^n\setminus \mathcal Y_{1,{\rm S}}^k,
\end{align}
the set of one-query output transcripts for which the adversary continues.

For $Q>1$, define recursively the set of output transcripts for which the adversary stops after the $Q$th query by
\begin{align}
    \mathcal Y_{Q,{\rm S}}^k
    &:=
    \left\{
    (\mathbf y_1,\ldots,\mathbf y_Q)
    \in
    \mathcal Y_{Q-1,{\rm C}}^k\times\{0,1\}^n:
    f_Q(\mathbf y_1,\ldots,\mathbf y_Q,k)={\rm S}
    \right\}.
\end{align}
Similarly, the set of output transcripts for which the adversary continues after the $Q$th query is
\begin{align}
    \mathcal Y_{Q,{\rm C}}^k
    &:=
    \left(
    \mathcal Y_{Q-1,{\rm C}}^k\times\{0,1\}^n
    \right)
    \setminus
    \mathcal Y_{Q,{\rm S}}^k .
\end{align}
We write
\begin{align}
    Y_{Q,{\rm S}}^k
    =
    |\mathcal Y_{Q,{\rm S}}^k|,
    \qquad
    Y_{Q,{\rm C}}^k
    =
    |\mathcal Y_{Q,{\rm C}}^k|,
\end{align}
where $|A|$ denotes the cardinality of a finite set $A$.

The recursive definitions above separate the adversary's behavior into two ingredients: the random oracle outputs and the deterministic stopping rules applied to the observed transcript.
The stopping rules only partition already possible transcripts into stopping and continuation sets; they do not change the distribution of the next oracle value.
The probabilistic input needed for the proof is therefore the following uniformity property of a random oracle at inputs that have not previously been queried.

\begin{lemma}[Uniformity of an unqueried oracle value]
\label{lemma:unqueried_uniformity}
Let $H:\Sigma\to\{0,1\}^n$ be uniformly distributed over all functions from $\Sigma$ to $\{0,1\}^n$.
Let $\mathbf x_1,\ldots,\mathbf x_r\in\Sigma$ be distinct inputs, and let $\mathbf y_1,\ldots,\mathbf y_r\in\{0,1\}^n$.
Then, conditioned on arbitrary values $H(\mathbf x_j)=\mathbf y_j$ of $H$ for $j\in\{1,\ldots,r-1\}$ on previously queried distinct inputs, the value of $H$ at any new input $\mathbf x_r$ remains uniformly distributed over $\{0,1\}^n$:
\begin{align}
    \Pr\!\left[
    H(\mathbf x_r)=\mathbf y_r
    \,\middle|\,
    H(\mathbf x_j)=\mathbf y_j\ \forall j=1,\ldots,r-1
    \right]
    =
    2^{-n}.
\end{align}

\end{lemma}

\begin{proof}
Let
\begin{align}
    E_{r-1}
    =
    \{H:H(\mathbf x_j)=\mathbf y_j\ \forall j=1,\ldots,r-1\}
\end{align}
and
\begin{align}
    E_r
    =
    \{H:H(\mathbf x_j)=\mathbf y_j\ \forall j=1,\ldots,r\}.
\end{align}
Since the inputs $\mathbf x_1,\ldots,\mathbf x_r$ are distinct, fixing the values of $H$ on $r-1$ inputs leaves the values of $H$ on the remaining $|\Sigma|-(r-1)$ inputs arbitrary. Hence
\begin{align}
    |E_{r-1}|=(2^n)^{|\Sigma|-(r-1)}.
\end{align}
Similarly, fixing the values of $H$ on all $r$ inputs leaves $|\Sigma|-r$ inputs arbitrary, so
\begin{align}
    |E_r|=(2^n)^{|\Sigma|-r}.
\end{align}
Because $H$ is uniformly distributed over all functions, conditioning on $E_{r-1}$ gives the uniform distribution over the functions in $E_{r-1}$. Therefore,
\begin{align}
    \Pr\!\left[
    H(\mathbf x_r)=\mathbf y_r
    \,\middle|\,
    H(\mathbf x_j)=\mathbf y_j\ \forall j=1,\ldots,r-1
    \right]
    =
    \frac{|E_r|}{|E_{r-1}|}
    =
    \frac{(2^n)^{|\Sigma|-r}}{(2^n)^{|\Sigma|-(r-1)}}
    =
    2^{-n}.
\end{align}
\end{proof}

The entropy $H_k[\mathbf Y]$ is determined by the probability distribution of the output transcript over the stopping sets $\mathcal Y_{Q,{\rm S}}^k$.
The following lemma characterizes the relevant transcript distributions.

\begin{figure}
    \centering
    \includegraphics[width=\linewidth]{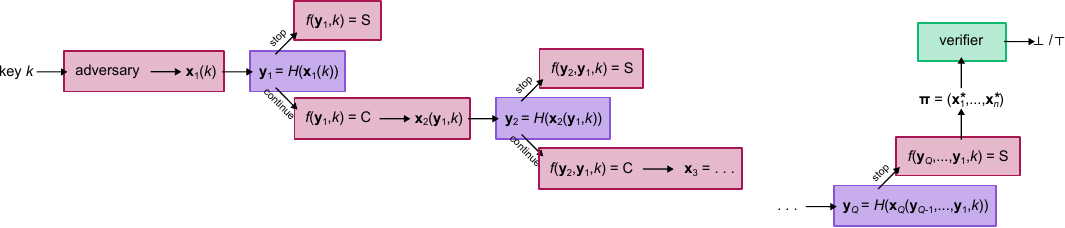}
    \caption{Decision tree for a classical adaptive algorithm.
    The adversary is instantiated with a key $k$.
    Based on $k$, the first input $\mathbf x_1(k)$ is generated and queried to the oracle, producing the output $\mathbf y_1=H(\mathbf x_1(k))$.
    The stopping rule $f_1(\mathbf y_1,k)$ determines whether the adversary stops or continues.
    If it continues, the next query input may depend on the previous output, e.g., $\mathbf x_2(\mathbf y_1,k)$.
    The algorithm proceeds until the first stopping event, after which the adversary outputs a proof string $\pi$ to the verifier.}
    \label{fig:tree}
\end{figure}

\begin{lemma}[Transcript distribution of a deterministic adaptive adversary]
\label{lemma:sublemma}
In the setting of Lemma~\ref{lemma:entropy_general}, let $Q^{\rm S}=Q$ denote the event that the adversary stops after exactly $Q$ nonredundant queries, and let $Q^{\rm C}=Q$ denote the event that the adversary continues after exactly $Q$ nonredundant queries.
Then
\begin{align}
\label{eq:PrQSQC|k_lemma}
    \Pr[Q^{\rm S}=Q|k]
    =
    \frac{Y_{Q,{\rm S}}^k}{2^{nQ}},
    \qquad
    \Pr[Q^{\rm C}=Q|k]
    =
    \frac{Y_{Q,{\rm C}}^k}{2^{nQ}}.
\end{align}
Moreover, conditioned on stopping after $Q$ queries, the output transcript is uniform over $\mathcal Y_{Q,{\rm S}}^k$:
\begin{align}
\label{eq:PrY1...YQ|QSk_lemma}
    \Pr[
    \mathbf y_1,\ldots,\mathbf y_Q
    \mid
    Q^{\rm S}=Q,k
    ]
    =
    \frac{1}{Y_{Q,{\rm S}}^k}
    \delta_{f_Q(\mathbf y_1,\ldots,\mathbf y_Q,k),{\rm S}}
    \prod_{j=1}^{Q-1}
    \delta_{f_j(\mathbf y_1,\ldots,\mathbf y_j,k),{\rm C}}.
\end{align}
Similarly, conditioned on continuing after $Q$ queries, the output transcript is uniform over $\mathcal Y_{Q,{\rm C}}^k$:
\begin{align}
\label{eq:PrY1...YQ|QCk_lemma}
    \Pr[
    \mathbf y_1,\ldots,\mathbf y_Q
    \mid
    Q^{\rm C}=Q,k
    ]
    =
    \frac{1}{Y_{Q,{\rm C}}^k}
    \prod_{j=1}^{Q}
    \delta_{f_j(\mathbf y_1,\ldots,\mathbf y_j,k),{\rm C}}.
\end{align}
\end{lemma}

\begin{proof}
We prove the lemma by induction on the number of nonredundant queries.
Throughout the proof, the key $k$ is fixed.
Since $k$ is sampled independently of the oracle, we have $\Pr[H|k]=\Pr[H]$, and $\Pr[H]$ is the uniform distribution over all functions $H:\Sigma\to\{0,1\}^n$.
There are $(2^n)^{|\Sigma|}$ such functions.

We first prove the base case $Q=1$.
Let $\mathbf x_1(k)$ be the first query.
For any $\mathbf y_1\in\{0,1\}^n$,
\begin{align}
    \Pr[\mathbf y_1|k]
    &=
    \sum_H
    \Pr[\mathbf y_1|H,k]\Pr[H|k] \\
    &=
    2^{-n|\Sigma|}
    \left|
    \{H:H(\mathbf x_1(k))=\mathbf y_1\}
    \right| \\
    &=
    2^{-n}.
\end{align}
Thus, the first output is uniform over $\{0,1\}^n$.
The stopping rule $f_1$ is a deterministic function of $\mathbf y_1$ and $k$.
It does not change the probability of any output before conditioning; it only partitions the uniformly distributed outputs into the stopping set $\mathcal Y_{1,{\rm S}}^k$ and the continuation set $\mathcal Y_{1,{\rm C}}^k$.
Therefore,
\begin{align}
    \Pr[Q^{\rm S}=1|k]
    &=
    \sum_{\mathbf y_1\in\mathcal Y_{1,{\rm S}}^k}
    \Pr[\mathbf y_1|k]
    =
    \frac{Y_{1,{\rm S}}^k}{2^n},\\
    \Pr[Q^{\rm C}=1|k]
    &=
    \sum_{\mathbf y_1\in\mathcal Y_{1,{\rm C}}^k}
    \Pr[\mathbf y_1|k]
    =
    \frac{Y_{1,{\rm C}}^k}{2^n}.
\end{align}
Conditioning a uniform distribution on either part of this deterministic partition leaves the distribution uniform on that part:
\begin{align}
    \Pr[\mathbf y_1|Q^{\rm S}=1,k]
    &=
    \frac{1}{Y_{1,{\rm S}}^k}
    \delta_{f_1(\mathbf y_1,k),{\rm S}},\\
    \Pr[\mathbf y_1|Q^{\rm C}=1,k]
    &=
    \frac{1}{Y_{1,{\rm C}}^k}
    \delta_{f_1(\mathbf y_1,k),{\rm C}}.
\end{align}
This proves the lemma for $Q=1$.

Now assume that the claims hold up to $Q-1$.
We first show that, conditioned on continuing after $Q-1$ queries, the next output is uniform before applying the next stopping rule.
Let
\begin{align}
    \tau_{Q-1}=(\mathbf y_1,\ldots,\mathbf y_{Q-1})\in\mathcal Y_{Q-1,{\rm C}}^k
\end{align}
be any continuation transcript.
The condition $\tau_{Q-1}\in\mathcal Y_{Q-1,{\rm C}}^k$ is equivalent to
\begin{align}
    f_j(\mathbf y_1,\ldots,\mathbf y_j,k)={\rm C}
    \quad\text{for all } j=1,\ldots,Q-1.
\end{align}
These conditions are deterministic predicates of the already observed transcript and the key.
Once the transcript $\tau_{Q-1}$ is fixed, they impose no further constraint on the value of the oracle at any unqueried input.

The next query input is a deterministic function
\begin{align}
    \mathbf x_Q(\tau_{Q-1},k)
    =
    \mathbf x_Q(\mathbf y_1,\ldots,\mathbf y_{Q-1},k)
\end{align}
of the previous transcript and the key.
Because repeated queries have been removed, this input has not appeared among the previous nonredundant queries.
Conditioned on $\tau_{Q-1}$, the values of $H$ on the previously queried inputs are fixed, but by Lemma~\ref{lemma:unqueried_uniformity}, the value of $H$ at the new input $\mathbf x_Q(\tau_{Q-1},k)$ remains uniformly distributed over $\{0,1\}^n$.
Therefore,
\begin{align}
\label{eq:next_output_uniform}
    \Pr[
    \mathbf y_Q
    |
    \tau_{Q-1},
    Q^{\rm C}=Q-1,k
    ]
    =
    2^{-n}.
\end{align}
Equivalently, before applying the $Q$th stopping rule, the extended transcript is uniform over
\begin{align}
    \mathcal Y_{Q-1,{\rm C}}^k\times\{0,1\}^n .
\end{align}
Indeed, using the induction hypothesis for the continuation transcript after $Q-1$ queries together with Eq.~\eqref{eq:next_output_uniform}, we obtain
\begin{align}
    \Pr[
    \mathbf y_1,\ldots,\mathbf y_Q
    |
    Q^{\rm C}=Q-1,k
    ]
    =
    \frac{1}{2^n Y_{Q-1,{\rm C}}^k}
    \prod_{j=1}^{Q-1}
    \delta_{f_j(\mathbf y_1,\ldots,\mathbf y_j,k),{\rm C}} .
\end{align}

The $Q$th stopping rule $f_Q$ is again only a deterministic predicate of the extended transcript $(\mathbf y_1,\ldots,\mathbf y_Q)$ and the key $k$.
It partitions the uniformly distributed set
\begin{align}
    \mathcal Y_{Q-1,{\rm C}}^k\times\{0,1\}^n
\end{align}
into the stopping set $\mathcal Y_{Q,{\rm S}}^k$ and the continuation set $\mathcal Y_{Q,{\rm C}}^k$.
Thus, applying $f_Q$ cannot bias the probabilities within either set; conditioning on stopping or continuing leaves the transcript uniform on the corresponding subset:
\begin{align}
\label{eq:Pr[y1...yQ,QS|QC,k]}
    \Pr[
    \mathbf y_1,\ldots,\mathbf y_Q, Q^{\rm S}=Q
    |
    Q^{\rm C}=Q-1,k
    ]
    = \frac{1}{2^n Y_{Q-1,{\rm C}}^k}
    \delta_{f_Q(\mathbf y_1,\ldots,\mathbf y_Q,k),{\rm S}} \prod_{j=1}^{Q-1}
    \delta_{f_j(\mathbf y_1,\ldots,\mathbf y_j,k),{\rm C}}
\end{align}
and analogously for continuation:
\begin{align}
\label{eq:Pr[y1...yQ,QC|QC,k]}
    \Pr[
    \mathbf y_1,\ldots,\mathbf y_Q, Q^{\rm C}=Q
    |
    Q^{\rm C}=Q-1,k
    ]
    = \frac{1}{2^n Y_{Q-1,{\rm C}}^k}
    \prod_{j=1}^{Q}
    \delta_{f_j(\mathbf y_1,\ldots,\mathbf y_j,k),{\rm C}}.
\end{align}

Quantitatively, for stopping, we have
\begin{align}
    \Pr[Q^{\rm S}=Q|k]
    &=
    \Pr[Q^{\rm C}=Q-1|k]\,
    \Pr[Q^{\rm S}=Q|Q^{\rm C}=Q-1,k] \\
    &\overset{\makebox[0pt]{\hspace{-0.1cm}\footnotesize\text{\eqref{eq:Pr[y1...yQ,QS|QC,k]}}}}{=}
    \frac{Y_{Q-1,{\rm C}}^k}{2^{n(Q-1)}}
    \sum_{\mathbf y_1,\ldots,\mathbf y_Q}
    \frac{1}{2^n Y_{Q-1,{\rm C}}^k}
    \prod_{j=1}^{Q-1}
    \delta_{f_j(\mathbf y_1,\ldots,\mathbf y_j,k),{\rm C}}\,
    \delta_{f_Q(\mathbf y_1,\ldots,\mathbf y_Q,k),{\rm S}} \\
    &=
    \frac{Y_{Q,{\rm S}}^k}{2^{nQ}}.
\end{align}
Similarly, for continuation,
\begin{align}
    \Pr[Q^{\rm C}=Q|k]
    &=
    \Pr[Q^{\rm C}=Q-1|k]\,
    \Pr[Q^{\rm C}=Q|Q^{\rm C}=Q-1,k] \\
    &\overset{\makebox[0pt]{\hspace{-0.1cm}\footnotesize\text{\eqref{eq:Pr[y1...yQ,QC|QC,k]}}}}{=}
    \frac{Y_{Q-1,{\rm C}}^k}{2^{n(Q-1)}}
    \sum_{\mathbf y_1,\ldots,\mathbf y_Q}
    \frac{1}{2^n Y_{Q-1,{\rm C}}^k}
    \prod_{j=1}^{Q}
    \delta_{f_j(\mathbf y_1,\ldots,\mathbf y_j,k),{\rm C}} \\
    &=
    \frac{Y_{Q,{\rm C}}^k}{2^{nQ}}.
\end{align}
This proves Eq.~\eqref{eq:PrQSQC|k_lemma} for $Q$.

It remains to prove the conditional uniformity statements.
For stopping after $Q$ queries, Bayes' rule, the calculation above and Eq.~\eqref{eq:Pr[y1...yQ,QS|QC,k]} give
\begin{align}
    \Pr[
    \mathbf y_1,\ldots,\mathbf y_Q
    |
    Q^{\rm S}=Q,k
    ]
    &=
    \frac{
    \Pr[
    \mathbf y_1,\ldots,\mathbf y_Q,Q^{\rm S}=Q|k
    ]
    }{
    \Pr[Q^{\rm S}=Q|k]
    }\\
    &=
    \frac{
    2^{-nQ}
    \prod_{j=1}^{Q-1}
    \delta_{f_j(\mathbf y_1,\ldots,\mathbf y_j,k),{\rm C}}\,
    \delta_{f_Q(\mathbf y_1,\ldots,\mathbf y_Q,k),{\rm S}}
    }{
    Y_{Q,{\rm S}}^k 2^{-nQ}
    }\\
    &=
    \frac{1}{Y_{Q,{\rm S}}^k}
    \prod_{j=1}^{Q-1}
    \delta_{f_j(\mathbf y_1,\ldots,\mathbf y_j,k),{\rm C}}\,
    \delta_{f_Q(\mathbf y_1,\ldots,\mathbf y_Q,k),{\rm S}} .
\end{align}
This is Eq.~\eqref{eq:PrY1...YQ|QSk_lemma}.
The proof for continuation is identical and relies on Eq.~\eqref{eq:Pr[y1...yQ,QC|QC,k]}:
\begin{align}
    \Pr[
    \mathbf y_1,\ldots,\mathbf y_Q
    |
    Q^{\rm C}=Q,k
    ]
    &=
    \frac{
    2^{-nQ}
    \prod_{j=1}^{Q}
    \delta_{f_j(\mathbf y_1,\ldots,\mathbf y_j,k),{\rm C}}
    }{
    Y_{Q,{\rm C}}^k 2^{-nQ}
    }\\
    &=
    \frac{1}{Y_{Q,{\rm C}}^k}
    \prod_{j=1}^{Q}
    \delta_{f_j(\mathbf y_1,\ldots,\mathbf y_j,k),{\rm C}},
\end{align}
which is Eq.~\eqref{eq:PrY1...YQ|QCk_lemma}.
This completes the induction and hence the proof of Lemma~\ref{lemma:sublemma}.
\end{proof}

Given Lemma~\ref{lemma:sublemma}, we can now prove Lemma~\ref{lemma:entropy_general}.

\begin{proof}[Proof of Lemma~\ref{lemma:entropy_general}]
We first evaluate the entropy of the output transcript conditioned on the event that the adversary stops after exactly $Q$ queries.
Equation~\eqref{eq:PrY1...YQ|QSk_lemma} of Lemma~\ref{lemma:sublemma} shows that, conditioned on $Q^{\rm S}=Q$, the output transcript
\begin{align}
    (\mathbf y_1,\ldots,\mathbf y_Q)\in\mathcal Y_{Q,{\rm S}}^k
\end{align}
is uniformly distributed over $\mathcal Y_{Q,{\rm S}}^k$.
Therefore,
\begin{align}
    H_k[\mathbf Y\mid Q^{\rm S}=Q]
    &=
    -\sum_{\mathbf y_1,\ldots,\mathbf y_Q}
    \Pr[\mathbf y_1,\ldots,\mathbf y_Q\mid Q^{\rm S}=Q,k]
    \log\!\left[
    \Pr[\mathbf y_1,\ldots,\mathbf y_Q\mid Q^{\rm S}=Q,k]
    \right] \\
    &=
    \log Y_{Q,{\rm S}}^k .
\end{align}
Using the identity
\begin{align}
    H[A\mid B]=\sum_b \Pr[B=b]H[A\mid B=b],
\end{align}
we obtain
\begin{align}
    H_k[\mathbf Y\mid Q^{\rm S}]
    &=
    \sum_{q\geq 1}
    \Pr[Q^{\rm S}=q\mid k]\,
    H_k[\mathbf Y\mid Q^{\rm S}=q] \\
    &=
    \sum_{q\geq 1}
    \frac{Y_{Q,{\rm S}}^k}{2^{nq}}
    \log Y_{Q,{\rm S}}^k,
\end{align}
where we used Eq.~\eqref{eq:PrQSQC|k_lemma}.
Finally, conditioning cannot increase Shannon entropy, so
\begin{align}
    H_k[\mathbf Y]\geq H_k[\mathbf Y\mid Q^{\rm S}].
\end{align}
Thus,
\begin{align}
    H_k[\mathbf Y]
    \geq
    \sum_{q\geq 1}
    \frac{Y_{Q,{\rm S}}^k}{2^{nq}}
    \log Y_{Q,{\rm S}}^k .
\end{align}
This proves Lemma~\ref{lemma:entropy_general}.
\end{proof}

\subsection{\label{SM:proof2}Proof of Eq.~\eqref{eq:HkY_main_lower} in the main text}

In this section, we prove Eq.~\eqref{eq:HkY_main_lower} from the main text.
We state the result in a slightly more detailed form as Lemma~\ref{lemma:entropy_lower_general}.

\begin{lemma}[General entropy lower bound]
\label{lemma:entropy_lower_general}
Consider the setting of Lemma~\ref{lemma:entropy_general}.
Let $Q^*>0$ and let $0<q<1$.
Set
\begin{align}
    Q_0=\lceil Q^*\rceil .
\end{align}
If
\begin{align}
    \Pr[Q^{\rm S}\geq Q^*\mid k]\geq q,
\end{align}
then
\begin{align}
    H_k[\mathbf Y]
    \geq
    q n Q_0\log[2]
    +
    q\log[q(1-2^{-n})].
\end{align}
In particular, since $Q_0\geq Q^*$,
\begin{align}
    H_k[\mathbf Y]
    \geq
    q n Q^*\log[2]
    +
    q\log[q(1-2^{-n})].
\end{align}
\end{lemma}

\begin{proof}
Since $Q^{\rm S}$ is an integer-valued random variable and $Q_0=\lceil Q^*\rceil$, the event $Q^{\rm S}\geq Q^*$ is equivalent to the event $Q^{\rm S}\geq Q_0$.
Using Eq.~\eqref{eq:PrQSQC|k_lemma} from Lemma~\ref{lemma:sublemma}, the assumed query lower bound can therefore be written as
\begin{align}
\label{eq:success_condition}
    \Pr[Q^{\rm S}\geq Q_0\mid k]
    &=
    \sum_{Q\geq Q_0}
    \Pr[Q^{\rm S}=Q\mid k] \\
    &=
    \sum_{Q\geq Q_0}
    \frac{Y_{Q,{\rm S}}^k}{2^{nQ}}
    \geq q .
\end{align}
By Lemma~\ref{lemma:entropy_general},
\begin{align}
    H_k[\mathbf Y]
    \geq
    \sum_{Q\geq 1}
    \frac{Y_{Q,{\rm S}}^k}{2^{nQ}}
    \log Y_{Q,{\rm S}}^k .
\end{align}
Dropping the nonnegative terms with $Q<Q_0$ gives
\begin{align}
\label{eq:H[Y | key k]_splitsum}
    H_k[\mathbf Y]
    \geq
    \sum_{Q\geq Q_0}
    \frac{Y_{Q,{\rm S}}^k}{2^{nQ}}
    \log Y_{Q,{\rm S}}^k ,
\end{align}
where the convention $0\log[0]=0$ is used.

We lower bound the right-hand side under the constraint~\eqref{eq:success_condition}.
Relaxing the integer variables $Y_{Q,{\rm S}}^k\in\mathbb Z_{\geq0}$ to real variables $Y_{Q,{\rm S}}^k\in\mathbb R_{\geq0}$ can only decrease the minimum, and hence gives a valid, possibly weaker, lower bound.
Fix
\begin{align}
\label{eq:lagrange_condition}
    \sum_{Q\geq Q_0}
    \frac{Y_{Q,{\rm S}}^k}{2^{nQ}}
    =
    p
\end{align}
with $p\geq q$.
For notational simplicity, write $Y_Q:=Y_{Q,{\rm S}}^k$.
We minimize
\begin{align}
    \sum_{Q\geq Q_0}
    \frac{Y_Q}{2^{nQ}}\log Y_Q
\end{align}
subject to Eq.~\eqref{eq:lagrange_condition}.
To this end, we minimize the Lagrange functional with multiplier $\mu$, given by
\begin{align}
    \mathcal L[\{Y_Q\}_Q,\mu]
    =
    \sum_{Q\geq Q_0}
    \frac{Y_Q}{2^{nQ}}\log Y_Q
    -
    \mu
    \left(
    \sum_{Q\geq Q_0}
    \frac{Y_Q}{2^{nQ}}
    -p
    \right).
\end{align}
For an interior optimum with $Y_Q>0$, differentiating with respect to $Y_Q$ gives
\begin{align}
    0
    =
    \frac{\partial \mathcal L}{\partial Y_Q}
    =
    2^{-nQ}(\log Y_Q+1-\mu),
\end{align}
and hence
\begin{align}
    Y_Q=e^{\mu-1}
\end{align}
is independent of $Q$.
Using the constraint~\eqref{eq:lagrange_condition}, we obtain
\begin{align}
    p
    &=
    \sum_{Q\geq Q_0}
    \frac{e^{\mu-1}}{2^{nQ}} \\
    &=
    e^{\mu-1}
    \frac{2^{-nQ_0}}{1-2^{-n}} .
\end{align}
Thus
\begin{align}
    Y_Q
    =
    e^{\mu-1}
    =
    p\,2^{nQ_0}(1-2^{-n}).
\end{align}
Substituting this value into Eq.~\eqref{eq:H[Y | key k]_splitsum} yields
\begin{align}
\label{eq:lower_bound_p}
    H_k[\mathbf Y]
    \geq
    p\log\!\left[
    p\,2^{nQ_0}(1-2^{-n})
    \right].
\end{align}

It remains to minimize this expression over $p\geq q$.
Let
\begin{align}
    g(p)
    =
    p\log\!\left[
    p\,2^{nQ_0}(1-2^{-n})
    \right].
\end{align}
The function $g$ is monotonically increasing whenever
\begin{align}
    p
    \geq
    \frac{2^{-nQ_0}}{e(1-2^{-n})}.
\end{align}
If
\begin{align}
    q
    \geq
    \frac{2^{-nQ_0}}{e(1-2^{-n})},
\end{align}
then the minimum of $g(p)$ over $p\geq q$ is attained at $p=q$, and Eq.~\eqref{eq:lower_bound_p} gives
\begin{align}
    H_k[\mathbf Y]
    \geq
    q\log\!\left[
    q\,2^{nQ_0}(1-2^{-n})
    \right].
\end{align}
If instead
\begin{align}
    q
    <
    \frac{2^{-nQ_0}}{e(1-2^{-n})},
\end{align}
then
\begin{align}
    q\log\!\left[
    q\,2^{nQ_0}(1-2^{-n})
    \right]<0.
\end{align}
Since entropy is nonnegative, $H_k[\mathbf Y]\geq0$, the same lower bound also holds in this case.
Consequently, for all $0<q<1$,
\begin{align}
    H_k[\mathbf Y]
    &\geq
    q\log\!\left[
    q\,2^{nQ_0}(1-2^{-n})
    \right] \\
    &=
    q n Q_0\log[2]
    +
    q\log[q(1-2^{-n})].
\end{align}
This proves the first claimed bound.
Since $Q_0\geq Q^*$, we further have
\begin{align}
    q n Q_0\log[2]
    +
    q\log[q(1-2^{-n})]
    \geq
    q n Q^*\log[2]
    +
    q\log[q(1-2^{-n})],
\end{align}
which proves the second bound.
\end{proof}

\subsection{\label{SM:proof_lowerbound}Proof of Eq.~\eqref{eq:explicit_entropy_schematic} in the main text}

To prove Eq.~\eqref{eq:explicit_entropy_schematic} in the main text, we combine the entropy bounds from Secs.~\ref{SM:proof1} and~\ref{SM:proof2} with the averaged soundness statement of the Yamakawa--Zhandry construction.
The point of the argument is that no fixed-key version of soundness is required.
Soundness averaged over the random oracle and the key already implies that any classical adversary with constant overall success probability must make at least the soundness query threshold with nonzero probability over the same randomness.
This averaged stopping-time lower bound, together with Jensen's inequality, yields the finite-size entropy bound used in the main text.

\begin{lemma}[Averaged minimal-query bound]
\label{lemma:averaged_minimal_queries_bound}
Let $\Pi$ be a $(Q_0,\varepsilon)$-sound proof of quantumness in the CROM.
Suppose that a keyed classical adversary $\mathcal A$ provides a correct proof with success probability at least $0<\Delta<1$, where the probability is taken over the random oracle and the key:
\begin{align}
    \Pr_{H,k}[{\rm succ}]\geq \Delta .
\end{align}
Let $Q^{\rm S}$ denote the number of nonredundant oracle queries made before $\mathcal A$ stops.
Then
\begin{align}
\label{eq:averaged_query_lower_bound}
    \Pr_{H,k}[Q^{\rm S}\geq Q_0]
    \geq
    \Delta-\varepsilon .
\end{align}
\end{lemma}

\begin{proof}
Consider the truncated adversary $\mathcal A_{<Q_0}$ that simulates $\mathcal A$ but aborts and outputs failure if $\mathcal A$ attempts to make the $Q_0$th nonredundant oracle query.
Thus, $\mathcal A_{<Q_0}$ makes at most $Q_0-1$ nonredundant queries, and hence at most $Q_0$ oracle queries.
By $(Q_0,\varepsilon)$-soundness,
\begin{align}
    \Pr_{H,k}[{\rm succ}\ \text{by}\ \mathcal A_{<Q_0}]
    \leq
    \varepsilon .
\end{align}
The truncated adversary succeeds exactly on those runs in which $\mathcal A$ succeeds and stops before making $Q_0$ nonredundant queries.
Therefore,
\begin{align}
    \Pr_{H,k}[{\rm succ}\wedge Q^{\rm S}<Q_0]
    \leq
    \varepsilon .
\end{align}
Using $\Pr_{H,k}[{\rm succ}]\geq\Delta$, we obtain
\begin{align}
    \Delta
    &\leq
    \Pr_{H,k}[{\rm succ}]\\
    &=
    \Pr_{H,k}[{\rm succ}\wedge Q^{\rm S}<Q_0]
    +
    \Pr_{H,k}[{\rm succ}\wedge Q^{\rm S}\geq Q_0]\\
    &\leq
    \varepsilon
    +
    \Pr_{H,k}[Q^{\rm S}\geq Q_0].
\end{align}
Rearranging proves Eq.~\eqref{eq:averaged_query_lower_bound}.
\end{proof}

The next lemma converts this averaged stopping-time constraint into an averaged entropy lower bound.
It is the averaged counterpart of Lemma~\ref{lemma:entropy_lower_general}.

\begin{lemma}[Entropy lower bound from an averaged stopping-time constraint]
\label{lemma:averaged_entropy_lower_bound}
Consider the setting of Lemma~\ref{lemma:entropy_general}.
Let $Q_0\geq 1$ be an integer and define
\begin{align}
    \bar q
    :=
    \Pr_{H,k}[Q^{\rm S}\geq Q_0].
\end{align}
Then
\begin{align}
\label{eq:averaged_entropy_lower_bound}
    H[\mathbf Y|k]
    \geq
    \bar q\, nQ_0\log[2]
    +
    \bar q\log[\bar q(1-2^{-n})],
\end{align}
where the right-hand side is interpreted by continuity at $\bar q=0$.
\end{lemma}

\begin{proof}
For each fixed key $k$, define
\begin{align}
    q_k
    :=
    \Pr[Q^{\rm S}\geq Q_0\mid k].
\end{align}
Applying Lemma~\ref{lemma:entropy_lower_general} with $q=q_k$ gives
\begin{align}
    H_k[\mathbf Y]
    \geq
    q_k nQ_0\log[2]
    +
    q_k\log[q_k(1-2^{-n})],
\end{align}
with the convention that the right-hand side is $0$ at $q_k=0$.
Averaging over keys yields
\begin{align}
    H[\mathbf Y|k]
    &=
    \sum_k \Pr[k]H_k[\mathbf Y]\\
    &\geq
    \sum_k\Pr[k]
    \left(
    q_k nQ_0\log[2]
    +
    q_k\log[q_k(1-2^{-n})]
    \right).
\end{align}
Let
\begin{align}
    g(q)
    :=
    q nQ_0\log[2]
    +
    q\log[q(1-2^{-n})],
    \qquad
    g(0):=0.
\end{align}
For $q>0$, we have
\begin{align}
    g''(q)=\frac{1}{q}>0,
\end{align}
so $g$ is convex on $q\geq0$ by continuity at $q=0$.
By Jensen's inequality,
\begin{align}
    \sum_k\Pr[k]g(q_k)
    \geq
    g\!\left(\sum_k\Pr[k]q_k\right).
\end{align}
Finally,
\begin{align}
    \sum_k\Pr[k]q_k
    =
    \Pr_{H,k}[Q^{\rm S}\geq Q_0]
    =
    \bar q.
\end{align}
This proves Eq.~\eqref{eq:averaged_entropy_lower_bound}.
\end{proof}

We now apply these general statements to the Yamakawa--Zhandry construction.

\begin{theorem}[Classical entropy lower bound for solving the Yamakawa--Zhandry problem]
\label{thm:entropy_lowerbound}
Let $\Pi$ be the proof-of-quantumness protocol from the Yamakawa--Zhandry construction~\cite{Yamakawa2024}.
Let $H:\Sigma\to\{0,1\}^n$ be a uniformly random oracle, and let $\mathcal A$ be a keyed classical adversary with unbounded computational time, unbounded memory, and an unbounded number of oracle queries.
Suppose that $\mathcal A$ succeeds in providing a valid proof with probability at least $\Delta>0$, independent of $\lambda$:
\begin{align}
\label{eq:case_weak}
    \Pr_{H,k}
    \left[
    \mathsf{Verify}^H(1^\lambda,k,\pi)=\top
    :
    \pi\xleftarrow{\$}\mathcal A^H(1^\lambda,k)
    \right]
    \geq
    \Delta .
\end{align}
Then, for sufficiently large $\lambda$, the entropy in the adversary's classical memory is lower bounded by
\begin{align}
\label{eq:explicit_entropy_schematic_sm}
    H[\mathbf X,\mathbf Y]
    \geq
    (\Delta-\varepsilon(\lambda))
    \left(
    (2^{\lfloor\log_2\lambda\rfloor}-1)2^{\lambda^c}\log[2]
    +
    \log\!\left[
    (\Delta-\varepsilon(\lambda))
    \left(1-2^{-(2^{\lfloor\log_2\lambda\rfloor}-1)}\right)
    \right]
    \right),
\end{align}
where
\begin{align}
\label{eq:eps(lambda)_Explicit}
    \varepsilon(\lambda)
    =
    L2^{-\lfloor\zeta n\rfloor}
    =
    \lambda^{2\lambda^{c'}}
    2^{-\lfloor\zeta(2^{\lfloor\log_2\lambda\rfloor}-1)\rfloor}
    =
    2^{2(\log_2\lambda)\lambda^{c'}
    -
    \lfloor\zeta(2^{\lfloor\log_2\lambda\rfloor}-1)\rfloor}.
\end{align}
In particular,
\begin{align}
    H[\mathbf X,\mathbf Y]
    \geq
    \Theta(\lambda 2^{\lambda^c}).
\end{align}
\end{theorem}

\begin{proof}
We first reduce the entropy of the adversary's full classical memory to the entropy of the oracle-output transcript.
Let
\begin{align}
    \Pr[\mathbf y_1,\mathbf y_2,\ldots]
    =
    \sum_{\mathbf x_1,\mathbf x_2,\ldots}
    \Pr[\mathbf x_1,\mathbf x_2,\ldots;\mathbf y_1,\mathbf y_2,\ldots]
\end{align}
be the marginal distribution of the output transcript.
Since marginalization cannot increase Shannon entropy,
\begin{align}
    H[\mathbf X,\mathbf Y]\geq H[\mathbf Y].
\end{align}
Moreover, conditioning on the key cannot increase entropy, so
\begin{align}
\label{eq:HY_sum_keys}
    H[\mathbf Y]\geq H[\mathbf Y|k]
    =
    \sum_k\Pr[k]H_k[\mathbf Y],
\end{align}
where
\begin{align}
    H_k[\mathbf Y]
    =
    -\sum_{\mathbf y_1,\mathbf y_2,\ldots}
    \Pr[\mathbf y_1,\mathbf y_2,\ldots|k]
    \log\!\left[
    \Pr[\mathbf y_1,\mathbf y_2,\ldots|k]
    \right].
\end{align}

We next use the averaged soundness of the Yamakawa--Zhandry construction.
For the folded Reed--Solomon parameters recalled in Sec.~\ref{sec:details_list_recoverability}, the protocol satisfies $(Q_0,\varepsilon(\lambda))$-soundness in the CROM with
\begin{align}
    Q_0=2^{\lambda^c}
\end{align}
and $\varepsilon(\lambda)$ given by Eq.~\eqref{eq:eps(lambda)_Explicit}.
By Lemma~\ref{lemma:averaged_minimal_queries_bound}, the success condition~\eqref{eq:case_weak} implies
\begin{align}
\label{eq:qbar_explicit}
    \Pr_{H,k}[Q^{\rm S}\geq Q_0]
    \geq
    \Delta-\varepsilon(\lambda).
\end{align}
Set
\begin{align}
    \bar q:=\Pr_{H,k}[Q^{\rm S}\geq Q_0].
\end{align}
Applying Lemma~\ref{lemma:averaged_entropy_lower_bound} gives
\begin{align}
\label{eq:HYk_avg_lower}
    H[\mathbf Y|k]
    \geq
    \bar q\, nQ_0\log[2]
    +
    \bar q\log[\bar q(1-2^{-n})].
\end{align}
For sufficiently large $\lambda$, the function
\begin{align}
    q\mapsto
    q nQ_0\log[2]
    +
    q\log[q(1-2^{-n})]
\end{align}
is monotonically increasing for all $q\geq \Delta-\varepsilon(\lambda)$.
Using Eq.~\eqref{eq:qbar_explicit}, we therefore obtain
\begin{align}
    H[\mathbf Y|k]
    \geq
    (\Delta-\varepsilon(\lambda))
    \left(
    nQ_0\log[2]
    +
    \log[
    (\Delta-\varepsilon(\lambda))(1-2^{-n})
    ]
    \right).
\end{align}
Combining this with Eq.~\eqref{eq:HY_sum_keys} and $H[\mathbf X,\mathbf Y]\geq H[\mathbf Y]$ yields
\begin{align}
    H[\mathbf X,\mathbf Y]
    \geq
    (\Delta-\varepsilon(\lambda))
    \left(
    nQ_0\log[2]
    +
    \log[
    (\Delta-\varepsilon(\lambda))(1-2^{-n})
    ]
    \right).
\end{align}
Finally, substituting
\begin{align}
    n=2^{\lfloor\log_2\lambda\rfloor}-1,
    \qquad
    Q_0=2^{\lambda^c}
\end{align}
gives Eq.~\eqref{eq:explicit_entropy_schematic_sm}.
Since $n=\Theta(\lambda)$, $Q_0=2^{\lambda^c}$, $\Delta>0$ is constant, and $\varepsilon(\lambda)=2^{-\Omega(\lambda)}$, the leading term scales as
\begin{align}
    H[\mathbf X,\mathbf Y]
    \geq
    \Theta(\lambda 2^{\lambda^c}).
\end{align}
This proves the theorem.
\end{proof}

\section{\label{SM:chernoff_proof}Chernoff bound for the experimental realization}

Here we provide the concentration bound used in the experimental protocol.
We assume that the verifier and the adversary repeat the experiment for $M$ rounds, each time with an independently chosen oracle instance and, possibly, a fresh key $k$.
The verifier accepts the success-rate condition if the adversary provides successful proofs in at least a fraction $\Delta+\eta$ of the rounds, where $\eta>0$ is a fixed margin.
Equivalently, if $X$ denotes the total number of successful rounds, the verifier requires
\begin{align}
    X\geq (\Delta+\eta)M .
\end{align}
If an integer threshold is needed, one may replace $(\Delta+\eta)M$ by $\lceil(\Delta+\eta)M\rceil$.

The purpose of the following calculation is to show that, if the true success probability of the adversary is smaller than $\Delta$, then the probability of passing this empirical success-rate test is exponentially small in $M$.
Informally, we want to bound
\begin{align}
    \Pr\!\left[
    X\geq(\Delta+\eta)M
    \,\middle|\,
    p_{\rm succ}<\Delta
    \right].
\end{align}

In a simplified model where the oracle instances used in different rounds are independent, each round can be modeled as an independent Bernoulli trial.
Let $X_1,\ldots,X_M$ be independent Bernoulli random variables with success probability $p$.
Thus $X_i=1$ if the $i$th round is successful and $X_i=0$ otherwise, with
\begin{align}
    \Pr[X_i=1]=p.
\end{align}
Let
\begin{align}
    X=\sum_{i=1}^M X_i
\end{align}
be the total number of successful rounds.

Fix $p<\Delta$.
For any $t>0$, Markov's inequality gives
\begin{align}
    \Pr[X\geq(\Delta+\eta)M]
    &=
    \Pr[e^{tX}\geq e^{t(\Delta+\eta)M}]\\
    &\leq
    e^{-t(\Delta+\eta)M}\operatorname{E}[e^{tX}].
\end{align}
The moment-generating function is
\begin{align}
    \operatorname{E}[e^{tX}]
    &=
    \prod_{i=1}^M \operatorname{E}[e^{tX_i}]\\
    &=
    (pe^t+1-p)^M\\
    &=
    \exp\!\left[M\log[pe^t+1-p]\right].
\end{align}
Therefore,
\begin{align}
    \Pr[X\geq(\Delta+\eta)M]
    \leq
    \inf_{t>0}
    \exp\!\left[
    M\log[pe^t+1-p]
    -t(\Delta+\eta)M
    \right].
\end{align}
Evaluating the infimum gives the standard Chernoff bound
\begin{align}
    \Pr[X\geq(\Delta+\eta)M]
    \leq
    \exp\!\left[
    -M D(\Delta+\eta\|p)
    \right],
\end{align}
where
\begin{align}
    D(a\|b)
    =
    a\log\!\left[\frac{a}{b}\right]
    +
    (1-a)\log\!\left[\frac{1-a}{1-b}\right]
\end{align}
is the binary Kullback--Leibler divergence.
Since $p<\Delta<\Delta+\eta$, the divergence $D(\Delta+\eta\|p)$ is minimized over $p<\Delta$ at $p=\Delta$.
Hence,
\begin{align}
    \Pr\!\left[
    X\geq(\Delta+\eta)M
    \,\middle|\,
    p_{\rm succ}<\Delta
    \right]
    \leq
    \exp\!\left[
    -M D(\Delta+\eta\|\Delta)
    \right].
\end{align}

\end{appendices}

\end{document}